\documentclass[12pt]{article} 
\usepackage{lmodern} 
\usepackage[utf8]{inputenc}
\usepackage{amsfonts,amsmath,amssymb,amsthm,latexsym}
\usepackage[a4paper]{geometry}
\usepackage{algorithm}

\newcommand{\N}{{\mathbb N}}

\newtheorem{prop}{Proposition}[section]
\newtheorem{lem}{Lemma}[section]

\theoremstyle{definition}
\newtheorem{ex}{Example}[section]
\newtheorem{rem}{Remark}[section]

\newcommand{\bi}{\begin{itemize}}
\newcommand{\ei}{\end{itemize}}
\newcommand{\bd}{\begin{description}}
\newcommand{\ed}{\end{description}}
\newcommand{\be}{\begin{enumerate}}
\newcommand{\ee}{\end{enumerate}}

\def\bc{\begin{center}}
\def\ec{\end{center}}

\def\no{\noindent}
\def\l{\left}
\def\r{\right}
\def\b{\big}

\def\m{\medskip}
\def\s{\smallskip} 
 
\begin{document}

\title{\bf Worst--Case Analysis\\
of Weber's GCD Algorithm}
\author{Christian Lavault\thanks{E-mail: \texttt{Christian.Lavault@lipn.univ-paris13.fr}}\ \ and\ 
S.~Mohamed Sedjelmaci\\[.5\baselineskip]
LIPN, CNRS UPRES-A 7030\\
Universit\'e Paris 13, F-93430 Villetaneuse}
\date{\empty}
\maketitle

\begin{abstract}
Recently, Ken Weber introduced an algorithm for finding the $(a,b)$-pairs satisfying $au+bv\equiv 0\pmod{k}$, with 
$0<|a|,|b|<\sqrt{k}$, where $(u,k)$ and $(v,k)$ are coprime. It is based on Sorenson's and Jebelean's ``$k$-ary reduction'' 
algorithms. We provide a formula for $N(k)$, the maximal number of iterations in the loop of Weber's GCD algorithm.

\s \no {\bf Keywords:} Integer greatest common divisor (GCD); Complexity analysis; Number theory.
\end{abstract}

\bibliographystyle{article}
\def\bibfmta#1#2#3#4{{#1}, {#2}, {\em #3}, #4.}
\bibliographystyle{book}
\def\bibfmtb#1#2#3#4{{#1}, {\em #2}, {#3}, #4.}

\section{Introduction}
The greatest common divisor (GCD) of integers $a$ and $b$, denoted by $\gcd(a,b)$, is the largest integer 
that divides both $a$ and $b$. 

Recently, Sorenson proposed the ``right-shift $k$-ary algorithm''~\cite{sor}. It is based on 
the following reduction. Given two positive integers $u>v$ relatively prime to $k$ (i.e., $(u,k)$ and $(v,k)$ are coprime), two integers $a,\;b$ can be found that satisfy
\begin{equation} \label{eq:ab}
au + bv\; \equiv \kern-.3cm \pmod{k}\ \quad \mbox{with} \quad 0 < |a|,\; |b| < \sqrt{k}.
\end{equation}
If we perform the transformation $(u,v)\longmapsto (u',v')$ (also called ``$k$-ary reduction''), where 
$(u',v') = \big(|au + bv|/k,\min(u,v)\big)$, which replaces $u$ with $u'=|au + bv|/k$, the size of $u$ 
is reduced by roughly $1/2\,\log_2(k)$ bits. Sorensen suggests table lookup to find sufficiently small $a$ and $b$ satisfying~(\ref{eq:ab}). By contrast, Jebelean~\cite{jeb1,jeb2} and Weber~\cite{web} both propose an 
easy algorithm, which finds such small $a$ and $b$ that satisfy~(\ref{eq:ab}) with time complexity $O(n^2)$, where $n$ represents the number of bits in the two inputs. 
This latter algorithm we call the ``Jebelean-Weber algorithm'', or {\em JWA} for short. 

\smallskip The present work focuses on the study of $N(k)$, the maximal number of  
iterations of the loop in {\em JWA}, in terms of $t = t(k,c)$ as a function 
of two coprime positive integers $c$ and $k$ $(0<c<k)$. Notice that this 
exact worst-case analysis of the loop does not provide the greatest lower 
bound on the complexity of {\em JWA}: it does not result in the optimality 
of the algorithm.

In the next Section~2, an upper bound on $N(k)$ is given, in Section~3, 
we show how to find explicit values of $N(k)$ for every  integer $k>0$. 
Section~4 is devoted to the determination of all integers $c>0$, 
which achieve the maximal value of $t(k,c)$ for every given $k>0$; 
that is the worst-case ocurrences of {\em JWA}. Section~5 contains 
concluding remarks.

\section{An Upper Bound on $N(k)$}
Let us recall the {\em JWA} as stated in~\cite{sel1,web}. The first instruction $c$ := $x/y\bmod k$ in {\em JWA} 
is not standard. It means that the algorithm finds $c\in [1,k-1]$, 
such that $cy=x+nk$, for some $n$ (where $x,y,k,c$, and $n$ are all integers).

\begin{algorithm}
\begin{tabbing}
\quad \=\quad  \kill
\> \\
Input: $x,y>0$, $k>1$, and \\
$\gcd(k,x)= \gcd(k,y)=1$. \\
Output: $(n,d)$ such that \\
$0<n,\; |d|<\sqrt{k}$, and
$ny\equiv dx\pmod{k}$. \\
\> $c := x/y\bmod k$ ; \\
\> $f_1=(n',d') := (k,0)$ ; \\
\> $f_2=(n'',d'') := (c,1)$ ; \\
{\bf while}\ $n''\ge \sqrt{k}$\ {\bf do} \\
\> $f_1 := f_1 - \lfloor n'/n'' \rfloor\,f_2$ ; \\
\> {\bf swap} $(f_1,f_2)$ \\
{\bf endwhile} \\
{\bf return} $f_2$
\end{tabbing}
\end{algorithm}
\no Notice that the loop invariant is $n'\, |d''| \;+\; n''\, |d'| \;=\; k$.  When $(n,d)$ is the output result of {\em JWA}, 
the pairs $(a,b) = (d,-n)$ and $(-d,n)$ meet property~(\ref{eq:ab}).

\subsection{Notation}
In {\em JWA}, the input data are the positive integers $k$, $u$ and $v$. However, for the purpose of the worst-case complexity analysis, we consider $c=u/v\bmod k$ in place of the pair $(u,v)$. 
Therefore, the actual input data of {\em JWA} are regarded as being $k$ and $c$, such that $0<c<k$, and $\gcd(k,c)=1$.

Throughout, we use the following notation. The sequence $(n_i,d_i)$ denotes the successive pairs produced 
by {\em JWA} when $k$ and $c$ are the input data. Let $t=t(k,c)$ denote the number of iterations 
of the loop of {\em JWA}; $t$ must satisfy the following inequalities:
\begin{equation} \label{eq:t}
n_t < \sqrt{k} < n_{t-1}\ \qquad \mbox{and}\ \qquad 0 < n_t,\, |d_t| < \sqrt{k}, 
\end{equation}
where finite sequence $D=(d_i)$~is defined recursively for  $i=-1,\;0\;,\;1,\ldots,\;(t-2)$ as
\begin{align} \label{eq:d}
d_{i+2} & =\; d_i \;-\; q_{i+2}\,d_i \quad \mbox{with}\ \quad d_{-1} = 0\ \ \mbox{and}\ \ d_0 = 1\nonumber\\
q_{i+2} & =\; \lfloor n_i/n_{i+1}\rfloor \qquad \mbox{with}\ \quad n_{-1} = k\ \ \mbox{and}\ \ n_0 = c. 
\end{align}
We denote by $Q=(q_i)$ the finite sequence of partial quotients defined 
in~(\ref{eq:d}). The sequence $D$ is uniquely determined from the choice 
of $Q$ (i.e., $D=D(Q)$), since the initial data $d_{-1}$ and $d_0$ 
are fixed and $D$ is an increasing function of the $q_i$'s in $Q$. 
Let $(F_n)$ $(n=0,1,\ldots)$ be the Fibonacci sequence, we define 
$m(k)$ by
$$m(k) =\; \max\l\{i\ge 0 \ |\ F_{i+1}\le \sqrt{k}\r\}\ \quad  \mbox{with}\ \ i\in \N.$$
For every given integer $k>0$, the maximal number of iterations of the loop of {\em JWA} is:
$$N(k) =\; \max\b\{t(k,c) \ |\  0 < c < k\ \ \mbox{and}\ \ \gcd(k,c) = 1\b\}.$$

\subsection{Bounding $N(k)$}

\begin{lem} \label{lem1}
With the above notation,
\begin{enumerate}
\item[(i)] $|d_t|\ge F_{t+1}$.

\item[(ii)] $N(k)\le m(k)$.
\end{enumerate}
\end{lem}
\begin{proof} \hfill

\no \emph{(i)} The proof is by induction on $t$.
\begin{itemize}
\item {\em Basis:} $|d_{-1}| = 0 = F_0$, $|d_0| = 1 = F_1$, and $|d_1| = q_1\ge 1 = F_2$.
\item {\em Induction step:} For every $i\ge 0$, suppose $|d_j|\ge F_{j+1}$ for $j = -1$, $0$, $1, \ldots$, $(i-1)$. 
Then,
\end{itemize}
$$|d_i| \;=\; |d_{i-2}| \;+\; q_i |d_{i-1}|\; \ge \;|d_{i-2}| \;+\; |d_{i-1}|\; \ge \;F_{i-1} + F_i \;=\; F_{i+1}$$
and \emph{(i)} holds.

\m \noindent \emph{(ii)} $F_{t+1}\le |d_t| < \sqrt{k}$. Hence $t = t(c,k)\;\le \;m(k)$, and also $N(k)\le m(k)$.
\end{proof}
Note that the following inequalities also hold
$$\phi^{m-1} \;<\; F_{m+1} \;le \;\sqrt{k} \;<\; F_{m+2} \;<\; \phi^{m+1},$$
where $\phi = (1+\sqrt{5})/2$ is the golden ratio. 

From Lemma~\ref{lem1} and the above inequalities, an explicit expression of $m(k)$ is easily derived,
$$m(k) = \lfloor \log_\phi(\sqrt{k})\rfloor\ \quad \mbox{or}\ \quad m(k) = \lceil \log_\phi(\sqrt{k})\rceil.$$

\begin{ex}

\begin{description}
\item For $k = 2^{10}$, $m(k) = 7$ and $t(k,633) = N(k)=m(k)=7$.
\item For $k=2^{16}$, $m(k)=12$ and $t(k,40,503)=N(k)=m(k)=12$.
\end{description}
\end{ex}
In the above examples, $N(k) = m(k)$. However, $N(k) < m(k)$ for some specific values of $k$; e.g. $k = 2^{12}$. 
(See Subsection~\ref{appli}, Case 1.)

\section{Worst-Case Analysis of \emph{JWA}}
In this section, we show how to find the largest number of iterations 
$N(k)$ for every integer $k>0$, and we exhibit all the values of $c$ corresponding to the worst case of {\em JWA}.

For $p\le m =m(k)$ and $c>0$ integer, let $I_p(k)$ and $J_p(k)$ be two sets defined as follows,
\begin{eqnarray*}
I_p(k) & = & \l\{c \;|\; (F_p/F_{p+1})k < c < (F_{p+1}/F_{p+2})k\r\}\ \ \text{for}\ \ p\ \text{even},\\
I_p(k) & = & \l\{c \;|\; (F_{p+1}/F_{p+2})k < c < (F_p/F_{p+1})k\r\}\ \ \text{for}\ \ p\ \text{odd}
\end{eqnarray*}
and
$$J_p(k) \;=\; I_p(k)\; \cap ;\{c\ |\ \gcd(k,c)=1\}.$$

\begin{prop} \label{prop1}
Let $k>9$ (i.e. $m(k)\ge 3$), and let c and n be two positive 
integers such that $\gcd(k,c)=1$ and $s\le m(k)=m$. The four following properties hold
\begin{enumerate}
\item[(i)] $c\in I_n(k)\ \Longrightarrow \; k/c = [1,1,\ldots,1,x]$, where $[1,1,\ldots,1,x]$ 
denotes a continued fraction having at least n times a ``$\/1$'' (including the leftmost $1$), 
and x is a sequence of positive integers (see e.g.~\cite{haw}).
\item[(ii)] If $J_{m-1}(k)\ne \emptyset$, then $N(k)=m$ or $m-1$.
\item[(iii)] If $J_{m-2}(k)\ne \emptyset$, then $N(k)=m$, $(m-1)$ or $(m-2)$.
\item[(iv)] If $k=2^s$, $N(k)=m$, $(m-1)$ or $(m-2)$.
\end{enumerate}
\end{prop}
\begin{proof}\hfill

\no \emph{(i)}~Let $a_n/b_n=[1,1,\ldots,1] = F_{n+1}/F_n$ be the $n$-th convergent of the golden ratio $\phi$, 
containing $n$ times the value ``1'' (see~\cite{haw,knu} for more details). To prove \emph{(i)}, 
we show that $F_{n+1}/F_n$ is the $n$-th convergent of the rational number $k/c$; in other words,
\begin{equation}
|(k/c) - (F_{n+1}/F_n)| \;<\; 1/(F_n)^2.
\end{equation}
Now, $(F_{n+1})^2 - F_nF_{n+2} = (-1)^n$ and, since $c\in I_n(k)$,
$$|(k/c) - (F_{n+1}/F_n)| \,<\, |(F_{n+1})^2 - F_n F_{n+2}|/(F_nF_{n+1}) \;=\; 1/(F_nF_{n+1}) \,<\, 1/(F_n)^2.$$

\no \emph{(ii)}~First, recall an invariant loop property which is also an Extended Euclidean Algorithm property: 
for $i = 1,\ldots,(t-1)$, where $t = t(k,c)$, we have that
\begin{equation} \label{eq:eea}
n_i\, |d_{i+1}| \;+\; n_{i+1}\, |d_i| \;=\; k. 
\end{equation}

We first prove that $n_{m-2} > \sqrt{k}$. In fact, if we assume that $J_{m-1}(k)\ne \emptyset$, then from {\em (i)}, 
there exists an integer $c$ such that $k/c = [1,1,\ldots,1,x]$ with $(m-1)$ such 1's. 
Then, $q_i = 1$ and $|d_i| = F_{i+1}$ for $i = 1,\ldots,(m-1)$.

\no Now if $n_{m-2} < \sqrt{k}$, then, since $n_{m-1} < n_{m-2}$,
\begin{eqnarray*}
k & = & n_{m-2}\, |d_{m-1}| \;+\; n_{m-1}\, |d_{m-2}| \;=\; n_{m-2}\, F_m \;+\; n_{m-1}\, F_{m-1}\\
& < & \sqrt{k}\, (F_m + F_{m-1}) = \sqrt{k}\, F_{m+1},
\end{eqnarray*}
and hence, $\sqrt{k} < F_{m+1}$, which contradicts the definition of $m(k)$, and $n_{m-2} > \sqrt{k}$.

\no If $n_{m-1} < \sqrt{k}$, then $t(k,c) = m-1$ and $N(k)\ge m-1$; else, if $n_{m-1} > \sqrt{k}$, then $N(k) = m$. 

\m \no \emph{(iii)}~The proof is similar to the previous one. There exists an integer $c$ such that 
$q_i = 1$ and $|di| = F_{i+1}$ for $i = 1,\ldots,(m-2)$. So, $n_{m-3} > \sqrt{k}$, and the result follows.

\m \no \emph{(iv)}~Let $\Delta_{m-2}$ be the size of the interval $I_{m-2}$. Then,
\begin{eqnarray*}
\Delta_{m-2} & = & |(F_{m-2}/F_{m-1})\, k \;-\; (F_{m-1}/F_m)\, k|\\
& = & k\, |F_{m-2}F_m - (F_{m-1})^2|/(F_{m-1}F_m) \;=\; k/(F_{m-1}F_m).
\end{eqnarray*}
Since 
$$2F_{m-1}F_m \,<\, (F_{m-1} + F_m)^2 = (F_{m+1})^2\ \ \text{and}\ \ (F_{m+1})^2\le k,\ \ \text{then}\ \ \Delta_{m-2} > 2.$$ 
Thus, out of two consecutive values within $I_{m-2}(k)$, at least one integer is odd. Therefore, $J_{m-2}(k)\ne \emptyset$ and we can apply \emph{(iii)}. (Note that this argument is not valid when $k$ is not a power of 2.)
\end{proof}

\begin{rem}\hfill

\begin{enumerate}
\item If $J_m(k)\ne \emptyset$, then $N(k)\ge m-1$, since $J_m(k)\subset J_{m-1}(k)\subset J_{m-2}(k)$.

\item The relation $N(k)=m-2$ holds for several $k$'s (e.g. for $k=90$).

\item For any given integer $k$, there may exists a positive integer $c$ such that $c\notin J_m(k)$, 
whereas $t(k,c)=m$. Such is the case when $k=15,849$: $m=10$, $I_m(k) = \{9,795\}$ 
and, since $\gcd(k,9,795)\ge 3$, $J_m(k)=\emptyset$. However, for $c = 11,468$, $t(k,11,468) = 10$.
\end{enumerate}
\end{rem}
This last example shows that $J_m(k)$ is not made of all integers $c$ such that $t(k,c) = m$, with $\gcd(k,c)=1$. 
Proposition~\ref{prop2} shows how to find all such numbers. For the purpose, two technical lemmas are needed first.

\begin{lem} \label{impl}
For every $m\ge 3$, the following three implications hold.
\begin{enumerate}
\item[(i)] $\exists i\ |\ q_i=2\ \Longrightarrow \ F_{m+1} + F_{m-1}\le |d_m|$.

\item[(ii)] $\exists i\ |\ q_i\ge 3\ \Longrightarrow \ |d_m|\ge F_{m+2} > \sqrt{k}$. 

\item[(iii)] $\exists i,j$ $(i\ne j)\ |\ q_i = q_j = 2\ \Longrightarrow \ |d_m|\ge F_{m+2} + 2F_{m-3} > \sqrt{k}$.
\end{enumerate}
\end{lem}
\begin{proof}\hfill

\no \emph{(i)}~Let $\Delta = \Delta(Q) = (\delta_i)_i$ be the sequence defined as: $\delta_{-1} = 0$, $\delta_0 = 1$, 
and $\delta_i = \delta_{i-2} + q_i \delta_{i-1}$, for $i = 1,\;2,\ldots,\;m$ with $Q = (1,2,1,\ldots,1)$.
An easy calculation yields $\delta_i = F_{i+1} + F_{i-1}$ for $i = 1,\;2,\ldots,\;m$.

\no On the other hand, let $(d_i)_i$ be a sequence satisfying~(\ref{eq:d}). 
We show that $|d_m|\ge \delta_m = F_{m+1} + F_{m-1}$ $(m\ge 3)$, and $\Delta$ is thus leading to the smallest possible 
$|d_m|$ satisfying the assumption in \emph{(i)}, i.e. $|d_m| = F_{m+1} + F_{m-1}$ $(m\ge 3)$. More precisely,
\begin{description}
\item If $D=D(Q)$ with $Q=(2,1,1,\ldots,1)$, then $|d_2| = 3$, $|d_3| = 5$, and $|d_m| = F_{m+2}$, 
whereas $\delta_2 = 3$, $\delta_3=4$ and $\delta_m = F_{m+1} + F_{m-1}$. Thus, $|d_m| > \delta_m$.

\item If $D=D(Q)$ with $Q=(1,1,\ldots,2,\ldots,1)$ and $q_p=2$ for some $p\ge 3$, 
then $|d_p| = F_{p-1} + 2F_p = F_{p+2}$ and $|d_{p+1}| = F_p + F_{p+2}$, whereas $\delta_p = F_{p+1} + F_{p-1}$ 
and $\delta_{p+1} = F_{p+2} + F_p$. It is then clear that $|d_i| > \delta_i$ for $i\ge p$,
and $|d_m|\ge \delta_m = F_{m+1} + F_{m-1}$.
\end{description}

\m \no \emph{(ii)}~Similarly, let $\Delta=\Delta(Q)$ defined by $Q=(1,3,1,\ldots,1)$, and let $D$ 
be a sequence satisfying the assumption. Then $|d_m|\ge \delta_m = F_{m+2}$ $(m\ge 3)$.
\begin{description}
\item If $D = D(Q)$ with $Q = (3,1,\ldots,1)$, then $|d_2| = 4$, $|d_3 | = 7$, whereas $\delta_2 = 4$ 
and $\delta_3 = 5$. Clearly, $|d_i| > \delta_i$ for $i=3$, and $|d_m| > \delta_m > F_{m+2}$.

\item If $D = D(Q)$ with $Q=(1,1,\ldots,3,\ldots,1)$ and $q_p=3$ 
for $p=3$, then $|d_p| = F_{p-1} + 3F_p = F_{p+3} + F_{p-2}$, and 
$|d_{p+1}| = F_{p+3} + F_p + F_{p-2}$, whereas $\delta_p = F_{p+2} + F_{p-3}$  and $\delta_{p+1} = F_{p+3} + F_{p-2}$.
Therefore, $|d_i|\ge \delta_i$ for $i\ge p$, and $|d_m|\ge \delta_m = F_{m+2} + F_{m-3} > F_{m+2}$.
\end{description}

\m \no \emph{(iii)}~The proof is similar to the previous one with $Q = (1,2,1,\ldots,1,2,1)$. For such a choice of $Q$, 
$|d_m|\ge \delta_m = F_{m+2} + 2F_{m-3}$, and the result follows. 
\end{proof}

\begin{lem} \label{dt}
For every $m\ge 3$, let $Q=(1,1,\ldots,1,2,1,\ldots,1)$, and let p be the index such that $q_p = 2$ $(q_j = 1$ for 
$j\ne p$, $1\le j\le m)$. Then, for $p=1,2,\ldots,m$, 
$|d_m|$ explicitly expresses as 
$$|d_m| \,=\, F_{m-p+1}\, F_{p+2} \;+\; F_{m-p}\, F_p.$$
\end{lem}
\begin{proof} The proof proceeds along the same arguments as for Lemma~\ref{impl}. \end{proof}

\begin{prop} \label{prop2}
For every integer $k\ge 9$ $(m\ge 3)$, if $t(k,c)=m$, then
\begin{description}
\item either $c\in J_m(k)$,

\item or $k/c = [1,\ldots,1,2,1,\ldots,1,x]$. That is, there exists $i\in \{1,\ldots,m\}$ such that $q_i = 2$ 
and for any $j\ne i$, $j\le m$ and $q_j = 1$.
\end{description}
In that case, the inequality $F_{m+1} + F_{m-1} < \sqrt{k}$ holds.
\end{prop}
\begin{proof} The proof follows from the inequalities~(\ref{eq:t}) and Lemma~\ref{impl}. \end{proof}

\subsection{Application of Proposition~\ref{prop2}} \label{appli}
The two following cases are examplified in Table 1. Assume $J_m(k) =\emptyset$.
\begin{description}
\item[Case 1:] $N(k)\le m(k) - 1$ holds, for example when $k=2^6$, $2^8$ or $2^{12}$, etc. (the inequality 
$F_{m+1} + F_{m-1} > \sqrt{k}$ holds).

\item[Case 2:] $N(k) = m(k)$. The procedure that determines all possible integers $c$ in the worst case is described 
in Section~\ref{wc}.
\end{description}

\section{Worst-Case Occurrences} \label{wc}
Assuming that $J_m(k) = \emptyset$, we search for the positive integers $c$ such that $t(k,c) = m(k)$.

\m \no {\bf Step 1.} Consider each value of $p$ $(p=1,2,\ldots,m)$, and select the $p$'s that meet 
the condition $|d_m| < \sqrt{k}$ (Lemma~\ref{impl} provides all values of $|d_m|$ for each such $m$). 
If $t(k,c)$ is still equal to $m$, then there exists a pair $(n_{m-1},n_m)$ satisfying the Diophantine equation
\begin{equation} \label{eq:dio}
n_{m-1}\, |d_m| \;+\; n_m\, |d_{m-1}| \;=\; k, 
\end{equation}
under the two conditions
\begin{equation} \label{eq:gcd}
\gcd(n_m,n_{m-1}) = 1\ \quad \mbox{with} \qquad n_m < \sqrt{k} < n_{m-1}
\end{equation}
and
\begin{equation} \label{eq:nt}
0 \;<\; n_m,\,|d_m| \;<\; \sqrt{k}. 
\end{equation}
The system of equations~(\ref{eq:dio})-(\ref{eq:gcd})-(\ref{eq:nt}) is denoted by $(\Sigma_Q)$, 
since it depends on $|d_m|$ and $|d_{m-1}|$, and thus on $Q$. Further, Eq.~(\ref{eq:dio}) is the expression 
of~(\ref{eq:eea}) when $i=m-1$, Eq.~(\ref{eq:nt}) expresses the exit test condition of {\em JWA} 
and Eq.~(\ref{eq:gcd}) ensures that $\gcd(k,c) = \gcd(n_m,n_{m-1}) = 1$.

\m \no {\bf Step 2.} Eq.~(\ref{eq:dio}) is solved modulo $|d_{m-1}|$. For $0\le a < |d_{m-1}|$,
$$n_{m-1}\;\equiv \;k/|d_m| \kern-.3cm \pmod{|d_{m-1}|}\;\equiv \;a \kern-.3cm \pmod{|d_{m-1}|},$$
and, from the inequality 
$$\sqrt{k} < n_{m-1} < k/|d_m|,$$
we have $n_{m-1} = a + r \,|d_{m-1}|$, where $r$ is a positive integer such that
$$(\sqrt{k} - a)/|d_{m-1}| \;<\; r \;<\; (k/|d_m| - a)/|d_{m-1}|.$$
Therefore, there exists only a finite number of solutions for $n_{m-1}$. Each solution of Eq.~(\ref{eq:dio}) (if any) 
fixes a positive integer $c\equiv n_{m-1}/|d_{m-1}|\pmod{k}$ such that $t(k,c)=m$, and $N(k)=m$.

\begin{ex}
Let $k=15,849$ and $m=10$. By Lemma~\ref{dt} (with $m=10$ and $p=2$), Eq.~(\ref{eq:dio}) yields 
$123 n_{m-1} + 76 n_m = 15,849$. Solving modulo 76 gives $n_{m-1}=127$ and $n_m=3$. 
The pair $(n_{m-1},n_m)$ corresponds to the value $c=11,468$, and $t(k,c) = N(k) = m(k) = 10$, while $J_m = \emptyset$.
\end{ex}

\subsection{Applications}
The following algorithm summarizes the results by computing the values of $N(k)$.

\newpage
\begin{algorithm}
\begin{tabbing}
\quad \=\quad \=\quad \kill
\> \\
\> $t := m$ ; \\
{\bf repeat} \\
\> {\bf if}\ $\exists c\in J_t | n_{t-1}>\sqrt{k}$\ 
{\bf then}\ $N := t$ \\
\> {\bf else}\ $\qquad \mbox{/*~} J_t=\emptyset$ or \\
\> no $c\in J_t$ satisfies $n_{t-1}>\sqrt{k}$~*/ \\
\> \> {\bf if}\ $(F_{t+1} + F_{t-1} < \sqrt{k})$\\
\> \> {\bf and}\ $(\exists c$ solution of $(\Sigma_Q))$\\
\> \> {\bf then}\ $N := t$\ {\bf else} $t := t - 1$ ;\\
{\bf until}\ $N$ is found
\end{tabbing}
\end{algorithm}

\begin{rem}\hfill

\begin{enumerate}
\item The algorithm terminates, since $N(k)\ge 1$ for every $k\ge 3$. 
Indeed, the first condition in the repeat loop always holds when $t=1$, since $k-1\in J_1(k)$ $(k\ge 3)$.

\item In the algorithm, $(\Sigma_Q)$ corresponds to the system~(\ref{eq:dio})-(\ref{eq:gcd})-(\ref{eq:nt}), 
where $t$ substitutes for $m$.
\end{enumerate}
\end{rem}
The case when $k > 1$ is an even power of 2 is of special importance, since it is related to the practical 
implementation of {\em JWA}~\cite{web}. Table 1 gives some of the values of $N(k)$, for $k=2^{2s}$  $(2\le s\le 16)$.

\begin{table}[h!]
\begin{center}
{\setlength{\tabcolsep}{5pt}
\begin{tabular}{|c|*{15}{l|}} 
\hline
$k$ & $2^4$ & $2^6$ & $2^8$ & $2^{10}$ & $2^{12}$ & $2^{14}$ & $2^{16}$ & $2^{18}$ %
& $2^{20}$ & $2^{22}$ & $2^{24}$ & $2^{26}$ & $2^{28}$ & $2^{30}$ & $2^{32}$ \\
\hline
$m(k)$ & 3 & 5 & 6 & 7 & 9 & 10 & 12 & 13 & 15 & 16 & 17 & 19 & 20 & 22 & 23 \\
$N(k)$ & 2 & 4 & 5 & 7 & 8 & 10 & 12 & 12 & 14 & 15 & 16 & 19 & 20 & 21 & 22 \\
\hline
\end{tabular}}
\caption{Values of $m(k)$ and $N(k)$ for $k = 2^{2s}$  $(2\le s\le 16)$.}
\end{center}
\end{table} 

\section{Concluding Remarks}
First we must point out that the condition $\gcd(k,c) = 1$ is a very strong requirement: it eliminates many integers 
within $I_m(k)$ and many solutions of $(\Sigma_Q)$. This can be seen e.g. when $k=2^{24}$. 
Then $m(k)=17$, and the choice of $Q=(1,2,1,\ldots,1)$, (i.e., $|d_m|=3,571$, $|d_{m-1}|=2,207$) yields 
$n_{m-1}=4,404$ and $n_m=476$, which leads to the solution $c=12,140,108$. We still have $t(k,c)=m(k)=17$ 
but unfortunately $\gcd(k,c)\ne 1$, and $N(k) = 16 =m(k) - 1$.

Checking whether $J_{m-2}(k)$ is empty is easy. It gives a straightforward answer to the question whether 
$$m(k)-2\; \le \;N(k)\;\le \;m(k)$$ 
or not.

\bigskip The following problems remain open.
\begin{itemize}
\item The example in Table~1 shows that, for $k=2^{2s}$ $(2\le s\le 16)$, the values of $N(k)$ are either 
$N(k) = m(k)$ or $N(k) = m(k)-1$. Does the inequality 
$$m(k)-1\;\le \;N(k)$$
always hold for $k = 2^{2s}$ $(n\ge 2)$? 

\item $N(k)$ is never less than $m(k)-2$. Are the inequalities 
$$m(k) - 2\le N(k)\le m(k)$$ 
true for every positive integer $k\ge 9$?

\item Find the greatest lower bound of $N(k)$ as a function of $m(k)$.
\end{itemize}

\vskip 1cm
\no {\small Christian {\sc Lavault}$^*$ and Sidi Mohamed {\sc Sedjelmaci}$^\dagger$ 
\newline LIPN, CNRS UPRES-A 7030 -- \texttt{http://lipn.univ-paris13.fr}
\newline Université Paris 13, 99 av. J.-B. Clément F-93430 Villetaneuse. 
\newline $^*$ \emph{E-mail:} \texttt{Christian.Lavault@lipn.univ-paris13.fr}, 
\newline \emph{URL:} \texttt{http://lipn.univ-paris13.fr/\textasciitilde lavault}
\newline $^\dagger$ \emph{E-mail:} \texttt{sms@lipn.univ-paris13.fr}
}

\begin{thebibliography}{99}

\bibitem{haw}\bibfmtb
{G.H.~Hardy and E.M.~Wright}{An Introduction to the Theory of Numbers}
{Oxford University Press}{London, 1979}

\bibitem{jeb1}\bibfmta
{T.~Jebelean}{A Generalization of the Binary GCD Algorithm}
{in Proc. of the Int. Symp. on Symbolic and Algebraic Computation (ISSAC'93)}{1993, 111-116}

\bibitem{jeb2}\bibfmta
{T.~Jebelean}{An Algorithm for Exact Division}{J. of Symbolic Computation}{15, 1993, 169-180}

\bibitem{knu}\bibfmtb
{D.E. Knuth}{The Art of Computer Programming:~seminumerical algorithms}
{Vol.~2, 2nd ed.}{Addisson Wesley, 1981}

\bibitem{sel1}\bibfmta
{M.S.~Sedjelmaci, C.~Lavault}{Improvements on the Accelerated Integer GCD Algorithm}
{Information Processing Letters}{61, 1997, 31-36}

\bibitem{sel2}\bibfmta
{M.S.~Sedjelmaci and C.~Lavault}{A new modular division algorithm 
and applications}{in Proc. of the 6th Int. Conf. on Theoretical Computer Science (ICTCS'98)}
{World Scientific, 1998, 65-76}

\bibitem{sor}\bibfmta
{J.~Sorenson}{Two Fast GCD Algorithms}{J. of Algorithms}{16, 1994, 110-144}

\bibitem{web}\bibfmta
{K.~Weber}{Parallel Implementation of the Accelerated Integer GCD Algorithm}
{J. of Symbolic Computation}{21, 1996, 457-466}

\end{thebibliography}
\end{document}